\documentclass[prl,twocolumn, english,reprint, longbibliography, superscriptaddress, breaklinks=true, showkeys, showpacs=false, nofootinbib]{revtex4}
\usepackage[T1]{fontenc}
\usepackage[latin9]{inputenc}
\usepackage{color}
\usepackage{babel}
\usepackage{amsmath}
\usepackage{amsthm}
\usepackage{amsfonts}
\usepackage{amssymb}
\usepackage{graphicx}
\usepackage{subfigure}
\usepackage{physics}
\usepackage{appendix}

\newtheorem{teo}{Theorem}

\usepackage[unicode=true,pdfusetitle,
 bookmarks=true,bookmarksnumbered=false,bookmarksopen=false,
 breaklinks=true,pdfborder={0 0 0},backref=false,colorlinks=true]
 {hyperref} 
\usepackage{breakurl}

\makeatletter
\@ifundefined{textcolor}{}
{%
 \definecolor{BLACK}{gray}{0}
 \definecolor{WHITE}{gray}{1}
 \definecolor{RED}{rgb}{1,0,0}
 \definecolor{GREEN}{rgb}{0,1,0}
 \definecolor{BLUE}{rgb}{0,0,1}
 \definecolor{CYAN}{cmyk}{1,0,0,0}
 \definecolor{MAGENTA}{cmyk}{0,1,0,0}
 \definecolor{YELLOW}{cmyk}{0,0,1,0}
}

\pdfoutput=1
\hypersetup{colorlinks=true,citecolor=blue,linkcolor=cyan,urlcolor=blue,filecolor= green, breaklinks=true}
\usepackage{url}
\usepackage{breakurl}

\makeatother

\begin{document}

\title{Entanglement monotones connect distinguishability and predictability}

\author{Marcos L. W. Basso}
\email{marcoslwbasso@hotmail.com}
\address{Departamento de F\'isica, Centro de Ci\^encias Naturais e Exatas, Universidade Federal de Santa Maria, Avenida Roraima 1000, Santa Maria, Rio Grande do Sul, 97105-900, Brazil}
\address{New adress: Centro de Ci\^encias Naturais e Humanas, Universidade Federal do ABC, Avenida dos Estados 5001, 09210-580 Santo Andr\'e, S\~ao Paulo, Brazil}
\author{Jonas Maziero}
\email{jonas.maziero@ufsm.br}
\address{Departamento de F\'isica, Centro de Ci\^encias Naturais e Exatas, Universidade Federal de Santa Maria, Avenida Roraima 1000, Santa Maria, Rio Grande do Sul, 97105-900, Brazil}

\selectlanguage{english}%

\begin{abstract} 
Distinguishability and predictability appear in different complementarity relations. Englert and Bergou pointed out the possible connection among distinguishability, predictability, and entanglement. They conjectured that an entanglement measure was hidden between the measures of distinguishability and predictability. Qureshi connected these quantities for a particular trio of measures. In this letter, we define a new entropic distinguishability measure and suggest an entanglement measure as the difference between it and an entropic predictability measure from the literature. An entanglement monotone is defined from the largest value of the distinguishability and the corresponding predictability, provided that the predictability satisfies the criteria already established in the literature. Our results formally connect an entanglement monotone with distinguishability and the corresponding predictability, without appealing to specific measures.

\end{abstract}


\maketitle


Bohr's complementarity principle was introduced as a qualitative statement about single quantum systems, or quantons \cite{Leblond}, which possess properties that are equally real but mutually exclusive \cite{Bohr}. The first quantitative version of the wave-particle duality was explored by Wootters and Zurek \cite{Zurek}. By considering unbalanced two-beam interferometers, in which the intensities of the beams were not
necessarily the same, Greenberger and Yasin defined an \textit{a priori} predictability $P$, for which the particle aspect is inferred once the quanton is more likely to follow one path than the other. This predictability limits the amount of visibility $V$ of the interference pattern according to the complementarity relation \cite{Yasin}
\begin{equation}
    P^2 + V^2 \le 1. \label{eq:cr1}
\end{equation}
By examining Eq. (\ref{eq:cr1}), one sees that even though an experiment can provide partial information about the wave and particle natures of a quantum system,
the more information it gives about one aspect of the system, the less information the experiment can provide about the other. Thus wave-particle duality becomes a quantitative statement. 

Following the line of reasoning of Ref. \cite{Zurek}, Englert \cite{Engle} considered a balanced two-beam interferometer in which one obtains the which-way information by introducing a path-detecting device, thus making the two paths distinguishable. The author derived the following inequality for such situations
\begin{align}
    D^2 + V^2 \le 1, \label{eq:cr2}
\end{align}
where $D$ is called distinguishability. In addition, Englert and Bergou \cite{Bethold} pointed out the possible connection among distinguishability, predictability, and entanglement, and conjectured that an entanglement measure was hidden  between the measures of distinguishability and predictability. Until now, many lines of thought were taken for quantifying the wave-particle properties of a quantum system \cite{Angelo, Bera, Coles, Hillery, Qureshi, Maziero, Lu}. Besides, D\"urr \cite{Durr} and Englert \textit{et al.} \cite{Englert} devised reasonable criteria for checking the reliability of newly defined predictability and visibility measures. As well, recently it was realized that the quantum coherence \cite{Baumgratz} would be a good generalization for the visibility of an interference pattern \cite{Bera, Bagan,  Mishra}. 

However, complementarity relations like the one in Eq. (\ref{eq:cr1}) are saturated only for pure, single-quanton, quantum states. And duality relations like Eq. (\ref{eq:cr2}) are saturated only if the joint state of quanton and path-detector is pure. So, if the quantum state of the quanton is mixed, 
it must be entangled with a path detector and/or with the environment. As stressed by Qian \textit{et al.} \cite{Qian}, complementarity inequalities like Eq. (\ref{eq:cr1}) do not really capture a balanced exchange between $P$ and $V$ because the inequality permits, for instance, that $P$ and $V$ decrease together due to the interaction of the system with its environment, allowing even the extreme case $P = V = 0$. Hence, something must be missing from Eq. (\ref{eq:cr1}). As noticed by Jakob and Bergou \cite{Janos}, this lack of local information about the system is due to entanglement, and a triality relation was derived:
\begin{align}
    P^2 + V^2 + \mathcal{C}^2 = 1, \label{eq:ccr1}
\end{align}
where $\mathcal{C}$ is the concurrence measure of entanglement \cite{Wootters}, which is recognized as the appropriate quantum correlation measure in a bipartite state of two qubits that completes the relation (\ref{eq:cr1}). Triality relations like Eq. (\ref{eq:ccr1}) are also known as complete complementarity relations \cite{Marcos}, since in Ref. \cite{Eberly} the authors interpreted this equality as completing the duality relation given by Eq. (\ref{eq:cr1}), thus turning the inequality into an equality.

Since Eq. (\ref{eq:cr2}) saturates for pure joint states of the quanton and path-detector, the authors in Ref. \cite{Jakob} realized that $D^2 = P^2 + \mathcal{C}^2$, thus confirming Englert and Bergou's conjecture for a specific trio of measures of distinguishability, predictability, and entanglement. This conjecture was recently discussed and confirmed for a different trio of measures in Ref. \cite{Tabish}.
As pointed out by the author, if the quanton couples with the path-detecting device, then of course entanglement is useful in relating distinguishability to predictability. However, if the quanton does not couple with the path-detecting apparatus, the concept of distinguishability will be useless, as one is not experimentally distinguishing between the paths. On the other hand, entanglement can be considered as an integral part of a triality relation since, for instance, it is possible that the quanton's path is entangled with an internal degree of freedom or with an inaccessible degree of freedom of the environment. Besides, it is worth emphasizing that these authors showed the connection between entanglement, distinguishability, and predictability for specific measures, which is not a general result.

In addition, it is worth pointing out that the distinguishability $D$ is referred by the authors in Ref. \cite{Bethold} as the which-alternative knowledge, which reflects the knowledge that the experimenter can learn from the path-detector device (or from measuring an observable of the environment). Rather, it is the largest value of $D$ which is called by the authors as the \textit{distinguishability} $\mathcal{D}$, and represents Nature's information about the actual alternative, i.e., the case where the quanton couples to a path detector device such that the states of the detectors become completely orthogonal to each other, representing an ideal von-Neumann measurement \cite{John}. Besides, the predictability $P$ reflects the experimentalist knowledge before measuring the path-detector device. Thus, the authors noticed the following hierarchy: $P \le D \le \mathcal{D}$. 

In this letter, 
we define a new entropic distinguishability measure and suggest an entanglement measure as the difference between this entropic distinguishability and an entropic predictability measure already defined in the literature \cite{Maziero}. Second, we prove that it is possible to define an entanglement monotone from the largest value of $D$ and the corresponding predictability, provided that the predictability satisfies the criteria already established in Refs. \cite{Durr, Englert}. Thus, this result formally connects an entanglement monotone with the largest value of $D$ and the corresponding predictability, without appealing to specific measures. 



Before stating the main result of this letter, let us first establish a new measure of distinguishability based on the von-Neumann entropy.  In Ref. \cite{Marcos}, we considered a quantum system described by the density operator $\rho_A$ of dimension $d_A$. The relative entropy of coherence of this state is defined as \cite{Baumgratz}
\begin{align}
    C_{re}(\rho_A) = \min_{\iota \in I} S_{vn}(\rho_A||\iota),
\end{align}
where $I$ is the set of all incoherent states, $S_{vn}(\rho_A||\iota) = \Tr(\rho_A \log_2 \rho_A - \rho_A \log_2 \iota)$ is the relative entropy, and $S_{vn}(\rho)$ denotes the von Neumann entropy of $\rho$. The minimization procedure leads to $\iota = \rho_{Adiag} = \sum_{i = 1}^{d_A} \rho^A_{ii} \ketbra{i}$. Thus 
\begin{align}
    C_{re}(\rho_A) = S_{vn}(\rho_{Adiag}) - S_{vn}(\rho_A) \label{eq:cre}.
\end{align}
Since $C_{re}(\rho_A) \le S_{vn}(\rho_{Adiag})$, it is possible to obtain an incomplete complementarity relation from this inequality:
\begin{equation}
    C_{re}(\rho_A) + P_{vn}(\rho_A) \le \log_2 d_A \label{eq:cr6},
\end{equation}
with $P_{vn}(\rho_A) := S^{\max}_{vn}- S_{vn}(\rho_{Adiag}) = \log_2 d_A + \sum_{i = 0}^{d_A - 1} \rho^A_{ii} \log_2 \rho^A_{ii}$ being a good measure of predictability, already defined in Ref. \cite{Maziero}, while $C_{re}(\rho_A)$ is a bone-fide measure of visibility. It is noteworthy that $P_{vn}(\rho_A)$ is defined as the difference between the maximum entropy and the \textit{a priori} entropy of the probability distribution defined by the diagonal elements of $\rho_A$.

Without loss of generality, in the context of $d$-slit interferometry, let $\ket{j}_A$ describe the state corresponding to the quanton taking the $j$-th path, then a general pure state is given by $\ket{\psi}_A = \sum_j a_j \ket{j}_A$, where $a_j$ represents the probability amplitude for the quanton to take the $j$-th path, and $\{\ket{j}_A\}_{j = 0}^{d_A - 1}$ can be regarded as an orthonormal path basis. Let us consider a path-detector which is capable of recording which path the quanton followed. In order to obtain the path-information from the detector, a measurement of a suitable observable $\mathcal{O}_B$ of the detector must be performed. Let $\{o_j\}_{j = 0}^{d_A - 1}$ be the eigenvalues of $\mathcal{O}_B$ with corresponding eigenstates $\{\ket{d_j}_B\}_{j = 0}^{d_A - 1}$, which are normalized but not necessarily orthogonal. The basic requirement for a quantum measurement, according to von Neumann \cite{John}, is to let the detector interact with a quanton and get entangled with it, i.e., $U(\ket{j}_A \otimes \ket{d_0}_B) \to \ket{j}_A \otimes \ket{d_j}_B$, where $\ket{d_0}$ is the initial detector state and $U$ represents the unitary evolution operator. Then, the state of the quanton and the detector is given by
\begin{equation}
    \ket{\Psi}_{A,B} = \sum_j a_j \ket{j}_A \otimes \ket{d_j}_B. \label{eq:psiab}
\end{equation}
When a measurement is performed, the outcome $o_k$ is obtained with probability $p_k = \Tr(I_A \otimes \Pi^k_B \rho_{AB}) \equiv \Tr(\Pi^k_B \rho_{AB})$, where  $\Pi^k_B = \ket{d_k}_B\bra{d_k}$ and $\rho_{AB} = \ket{\Psi}_{A,B}\bra{\Psi}$. The reduced density operator of the quanton, which is conditioned on the outcome $o_k$ of the observable $\mathcal{O}_B$, is given by $\rho^{(k)}_A = p^{-1}_k \Tr_B(\Pi^k_B \rho_{AB})$. Thus, after the non-selective measurement of $\mathcal{O}_B$, the reduced state of the quanton is described by the statistical ensemble $\rho_A = \sum_k p_k \rho^{(k)}_A$. This ensemble is sorted into sub-ensembles $\rho^{(k)}_A$ depending on the measurement outcome of the detector observable. The sorting into the sub-ensembles depends on the choice of the observable $\mathcal{O}_B$, as described in Ref. \cite{Bethold}. Let us define the distinguishability of the sub-ensemble $\rho^{(k)}_A$ as $D^k_{vn}(\rho^{(k)}_A) := \log_2 d_A - S_{vn}(\rho^{(k)}_{A diag})$. The complete which-way knowledge is defined by the statistical average over all possible outcomes $o_k$, i.e., $D_{vn}(\rho_A) := \sum_k p_k D^k_{vn}(\rho^{(k)}_A) = \log_2 d_A - \sum_k p_k S_{vn}(\rho^{(k)}_{A diag})$. Similarly, it is possible to define the quantum coherence of the sub-ensemble $\rho^{(k)}_A$ as $\mathcal{C}^k_{re}(\rho^{(k)}_A) := S_{vn}(\rho^{(k)}_{A diag}) - S_{vn}(\rho^{(k)}_{A})$, such that the averaged quantum coherence for the sorting based on the measurement of $\mathcal{O}_B$ is given by $\mathcal{C}_{re}(\rho_A)  = \sum_k p_k
\mathcal{C}^k_{re}(\rho^{(k)}_A)$. Therefore, we have the following complementarity relation
\begin{align}
    D_{vn}(\rho_A) + \mathcal{C}_{re}(\rho_A)  \le \log_2 d_A, \label{eq:cr3}
\end{align}
which relates the distinguishability (or which-alternative knowledge) and the average quantum coherence of sorting $\rho_A$ into sub-ensembles $\rho^{(k)}_A$ depending on the measurement outcome of the detector observable, and holds for all $\mathcal{O}_{B}$. This complementarity relation is equivalent to Eq. (39) of Ref. \cite{Bethold}, while Eq. (\ref{eq:cr6}) is equivalent to Eq. (36) of Ref. \cite{Bethold}. 

It is noteworthy that $D_{vn}(\rho_A)$ can be rewritten as
\begin{align}
    D_{vn}(\rho_A)  =& \ln d_A - S_{vn}(\rho_{A diag}) \nonumber \\ & + S_{vn}(\rho_{A diag}) - \sum_k p_k S_{vn}(\rho^{(k)}_{A diag}) \nonumber \\
     =& P_{vn}(\rho_A) + E(\rho_{A}),
\end{align}
where $E(\rho_A) :=  S_{vn}(\rho_{A diag}) - \sum_k p_k S_{vn}(\rho^{(k)}_{A diag})$. It is easy to see that this quantity vanishes if and only if all $\ket{d_j}_B$ are identical, i.e., $E(\rho_A)$ is null if and only if the bipartite quantum state $\ket{\Psi}_{A,B}$ is separable. Besides, $E(\rho_A) \ge 0$ since $S_{vn}(\rho_{A diag})$ is a concave function. On the other hand, if  $\braket{d_k}{d_j} = \delta_{jk}\ \forall j,k$ and $\abs{a_i} = 1/\sqrt{d_A}$, which means the state is maximally entangled, then $E(\rho_A) = S_{vn}(\rho_{A diag}) = S_{vn}(\rho_{A}) = \log_2 d_A$. Therefore, $E(\rho_A)$ can be considered as a indicative of entanglement in this situation. In addition, Englert and Bergou conjectured, in Ref. \cite{Bethold}, that $\mathcal{E}(\rho_{A}) := (D_{vn}(\rho_A) + \mathcal{C}_{re}(\rho_A)) - (P_{vn}(\rho_A) + C_{re}(\rho_A))$ is indicative of entanglement, and that perhaps one could introduce a useful quantitative measure for the entanglement by studying the properties of this difference. Here $\mathcal{E}(\rho_{A}) =  S_{vn}(\rho_{A}) - \sum_k p_k S_{vn}(\rho^{(k)}_{A}),$ which has the same properties as $E(\rho_A)$.

Now, we turn to the main result of this letter. By noting that the largest value $\mathcal{D}_{vn}$ of the distinguishability $D_{vn}$ is by definition $\mathcal{D}_{vn}:=\max_{\Pi^k_B}D_{vn} = \log_2 d_A$, since the maximum is reached when the states $\{\ket{d_j}_B\}_{j = 1}^{d_A}$ are orthogonal and therefore $S_{vn}(\rho^{(k)}_{A diag}) = 0$. In this case, $\rho_{A diag} = \rho_A$ and thus $E(\rho_A) = \mathcal{E}(\rho_{A}) = S_{vn}(\rho_A)$, which is a well known entanglement monotone for bipartite pure states \cite{Vedral}. Before stating the main theorem, let us present a quick review on entanglement monotones, which will be useful to prove the main theorem. Following Ref. \cite{Zhu}, let us denote $\mathcal{D}(\mathcal{H)}$ the set of density matrices on $\mathcal{H} \simeq \mathbb{C}^d$ and $U(d)$ the group of unitary operators on $\mathcal{H}$. In addition, let $\mathcal{F}_U$ be the set of local unitarily invariant functions on $\mathcal{D}(\mathcal{H)}$ such that each function $f \in \mathcal{F}_U$ is defined on the space of density matrices for each positive integer $d = \dim \mathcal{H}$. For any given $d$, the function satisfies
\begin{align}
    f(U\rho U^{\dagger}) = f(\rho), \ \ \forall \rho \in \mathcal{D}(\mathcal{H)}, \ U \in U(d). \label{eq:unit}
\end{align}
Therefore $f(\rho)$ can be taken as a function of the eigenvalues of $\rho$. By restricting ourselves to the set $\mathcal{F}_{Uc} \subset \mathcal{F}_{U}$ of local unitarily invariant and real concave functions on $\mathcal{D}(\mathcal{H)}$, then each function $f \in \mathcal{F}_{Uc}$ satisfies Eq. (\ref{eq:unit}) and
\begin{align}
    f(\lambda \rho + & (1 - \lambda) \sigma) \ge \lambda f(\rho) + (1 - \lambda) f(\sigma) \nonumber \\
    & \forall \rho, \sigma \in \mathcal{D}(\mathcal{H)}, \lambda \in[0,1],
\end{align}
for any given $d$. Now, let $\mathcal{H} \simeq \mathcal{H}_A \otimes \mathcal{H}_B$ be a bipartite Hilbert space of a bipartite quantum system $A$ and $B$ with dimension $d_A = d_B = d$. The fact that the dimensions of the subsystems are the same is not essential here. Any function $f \in \mathcal{F}_{Uc} $ can be used to construct an entanglement monotone $E_f$ on $\mathcal{D}(\mathcal{H)}$ as follows. For a pure state $\ket{\Psi} \in \mathcal{H}$,
\begin{align}
    E_f(\Psi) := f(\Tr_B(\ketbra{\Psi})) = f(\rho_A). \label{eq:pure}
\end{align}
Then, it is possible to extend the monotone for mixed states $\rho \in \mathcal{D}(\mathcal{H)}$ by the convex roof construction:
\begin{align}
    E_f(\rho) := \min_{\{p_j, \ket{\Psi_j}\}} \sum_j p_j E_f(\Psi_j), \label{eq:mixed}
\end{align}
where the minimization runs over all pure state ensembles $\{p_j, \ket{\Psi_j}\}$ for which $\rho = \sum_j p_j \ketbra{\Psi_j}$. Conversely, the restriction to pure states of any entanglement monotone is identical to $E_f$ for a given $f \in \mathcal{F}_{Uc}$. These facts were stated as a theorem in Ref. \cite{Vidal}.

Besides that, let $\Delta_d$ be the probability simplex of probability vectors with $d$ components. A function on $\Delta_d$ is symmetric if it is invariant under permutations of the components of probability vectors. Let $\mathcal{F}_s$ be the set of symmetric functions on the probability simplex such that each function $f \in \mathcal{F}_s$ is defined for any given positive integer $d$. The authors in Ref. \cite{Zhu} showed that any symmetric function  $f \in \mathcal{F}_s$  can be lifted to an unitarily invariant function on $\mathcal{D}(\mathcal{H)}$: $\hat{f}(\rho) := f(\text{eig}(\rho)), \ \forall \rho \in \mathcal{D}(\mathcal{H)},$ where $\text{eig}(\rho)$ are the eigenvalues of $\rho$. Conversely, any unitarily invariant function $f$ on the space of density matrices defines a symmetric function on the probability simplex when restricted to diagonal density matrices:$\check{f}(p) := f(\rho_{diag}), \ \forall p \in \Delta_d,$
where $p\in \Delta_d$ represents a probability distribution in $\Delta_d$, which in this case is given by $\rho_{diag}$.  Therefore, for any concave function $f \in  \mathcal{F}_s$, $E_f$ defined by Eqs. (\ref{eq:pure}) and (\ref{eq:mixed}) is an entanglement monotone. Conversely, the restriction to pure states of any entanglement monotone is identical to $E_f$ for a given concave function $f \in  \mathcal{F}_s$. This claim was proved in Ref. \cite{Zhu}.

Therefore, we can see why $E_{vn}(\rho_A)$, when restricted to the largest value of $D_{vn}(\rho_A)$, is an entanglement monotone. First, $E_{vn}(\rho_A) = \mathcal{D}_{vn} - P_{vn}(\rho_{A}) = f(\Tr_B(\ket{\Psi}_{A,B}\ket{\Psi}))$. Second, $P_{vn}(\rho_A)$ is a convex functions of $\rho_A$, thus  $E_{vn}(\rho_A)$ is concave \cite{Roberts}. In addition, $P_{vn}(\rho_A)$ is invariant under the permutation of the elements $\rho_{jj}$ by the criteria $C2$ of the Appendix \ref{sec:appe}. Therefore, we can use any largest value $\mathcal{D}$ of a distinguishability measure used in complementarity relations of the type (\ref{eq:cr3}) together with a corresponding predictability measure used in complementarity relations of the type (\ref{eq:cr6}) to obtain entanglement monotones $E_f$ defined by Eqs. (\ref{eq:pure}) and (\ref{eq:mixed}), provided that the measures of predictability satisfy the criteria established in Refs. \cite{Durr, Englert}, and stated in Appendix \ref{sec:appe}. Therefore, we have the following theorem.

\begin{teo}
Let $\mathcal{D}$ be the largest value of any distinguishability measure used in complementarity relations of the type (\ref{eq:cr3}), together with a corresponding predictability measure $P$ used in complementarity relations of the type (\ref{eq:cr6}), which saturates only for pure quantum states. Therefore, the quantity
\begin{equation}
    E_f := \mathcal{D} - P(\rho_A) \label{eq:entmon_}
\end{equation}
is an entanglement monotone as defined by Eqs. (\ref{eq:pure}) and (\ref{eq:mixed}), provided that the measures of predictability satisfy the criteria established in the literature and stated in Appendix \ref{sec:appe}.
\end{teo}

\begin{proof}
Since the measures $P(\rho_A)$ satisfy the criteria established in Refs. \cite{Durr, Englert}, then $P(\rho_A)$ is a convex function, which implies that $E_f := \mathcal{D} - P(\rho_A)$ is a concave function, since the largest value $\mathcal{D}$ of a distinguishability measure is just a constant. Now, let $\ket{\Psi}_{A,B} \in \mathcal{H}_A \otimes \mathcal{H}_B$ be a purification of $\rho_A$, i.e., $\rho_A = \Tr_B \ket{\Psi}_{A,B}\bra{\Psi}$. Using the Schmidt decomposition $\ket{\Psi} = \sum_k \sqrt{\lambda_k}\ket{\phi_k}_A \otimes \ket{\psi_k}_B$, we can write $\rho_A =  \sum_k \lambda_k \ketbra{\phi_k}$, which implies that $P(\rho_A) \ge 0$. But $P(\rho_A)$ must be invariant under permutations of the states' indexes, which implies that $P(\rho_A)$ is invariant under the permutations of the components of the probability vectors $\vec{\lambda}=(\lambda_{0},\cdots,\lambda_{d-1})$, and therefore $E_f = \mathcal{D} - P(\rho_A) $ is invariant under the permutations of the components of $\vec{\lambda}$.
\end{proof}

An equivalent but normalized definition of Eq. (\ref{eq:entmon_}) is given by $E_f:= 1 - \frac{1}{\mathcal{D}}P(\rho_A)$. The largest value $\mathcal{D}$ is the constant that saturates complementarity relations of type of those in Eqs. (\ref{eq:cr6}) and (\ref{eq:cr3}). Besides, the theorem above remains true for the conjecture $\mathcal{E}_f := (\mathcal{D} + \mathcal{C}(\rho_A)) - (P(\rho_A) + C(\rho_A))$, since the largest value $\mathcal{D}$ of $D$ is obtained when the states of the path detector are orthogonal, which implies that $\mathcal{C}(\rho_A) = C(\rho_A) = 0$. In addition, from Eq. (\ref{eq:entmon_}), it is straightforward to see that all states that maximize $E_f$ has the same form, i.e., these states are those in which the reduced states satisfy $P(\rho_A) + C(\rho_A) = 0$, where $C(\rho_A)$ is the corresponding quantum coherence measure. From the criteria in Appendix \ref{sec:appe}, it is easy to see that these reduced states are the ones that are maximally mixed. 

As an example, let us consider the measure of entanglement obtained by Qureshi through the difference between a distinguishability measure $D_{\mathcal{Q}}$, given by Eq. (13) in Ref. \cite{Tabish}, and the corresponding predictability measure $P_{\mathcal{Q}}$, given by Eq. (11) in Ref. \cite{Tabish}, i.e.,
\begin{align}
    \mathcal{E}_{\mathcal{Q}} = \frac{1}{d_{A}-1}\sum_{j \neq k}\Big( \sqrt{\rho_{jj}\rho_{kk}} - \sqrt{\rho_{jj}\rho_{kk}}\abs{\braket{d_j}{d_k}}\Big), 
\end{align}
where $\rho_{jk} = a_j a^*_k$ of Eq. (\ref{eq:psiab}) and $\ket{d_j}$ are the states of the detector. As well, $D_{\mathcal{Q}}$ is maximized when the states of the detector are orthogonal to each other. Therefore, in this case $\mathcal{E}_{\mathcal{Q}} = \frac{1}{d_{A}-1} \sum_{j \neq k} \sqrt{\lambda_j \lambda_k}$, where $\lambda_j$ are the coefficients of the Schmidt decomposition. It is worth pointing out that the author suggested that a more rigorous analysis is needed to consider $\mathcal{E}_{\mathcal{Q}}$ as a proper entanglement measure. Thus, one can see that, when restricted to the largest value of $D_{\mathcal{Q}}$, $\mathcal{E}_{\mathcal{Q}}$ is an entanglement monotone. In fact, $\mathcal{E}_{\mathcal{Q}} = (d_{A}-1) \mathcal{R}$, where $\mathcal{R}$ is the robustness of entanglement for global pure states \cite{Tarrach}.

Summing up, in this letter we pushed foward the conjecture of Englert and Bergou.
We started defining a new entropic distinguishability measure and suggesting a corresponding entanglement measure as the difference between this entropic distinguishability and an entropic predictability measure already defined in the literature \cite{Maziero}. More importantly, we proved that it is possible to define an entanglement monotone from the largest value of $D$ and the corresponding predictability, provided that the predictability satisfies the criteria already established in Refs. \cite{Durr, Englert}. Thus, this result formally connects an entanglement monotone with the largest value of $D$ and the corresponding predictability, without appealing to specific measures. In addition, it opens the possibility for establishing new entanglement measures whenever a new distinguishability measure, and a corresponding predictability measure that satisfies the criteria of Refs. \cite{Durr, Englert}, is obtained.
  
\begin{acknowledgments}
This work was supported by the Coordena\c{c}\~ao de Aperfei\c{c}oamento de Pessoal de N\'ivel Superior (CAPES), process 88882.427924/2019-01, and by the Instituto Nacional de Ci\^encia e Tecnologia de Informa\c{c}\~ao Qu\^antica (INCT-IQ), process 465469/2014-0.
\end{acknowledgments}

\appendix

\section{Criteria for Predictability/Visibility Measures}
\label{sec:appe}
D\"urr \cite{Durr} and Englert \textit{et al.} \cite{Englert} established criteria that can be taken as a standard for checking for the reliability of newly defined predictability measures $P(\rho)$ and interference pattern visibility quantifiers $V(\rho)$. These required properties can be stated as follows:
\begin{itemize}
\item[C1] $P$ must be a continuous function of the diagonal elements of the density matrix. And $V$ must be a continuous function of the elements of the density matrix.
\item[C2] $P$ and $V$ must be invariant under permutations of the base states indexes.
\item[C3] If $\rho_{jj}=1$ for some $j$, then $P$ must reach its maximum value, while $V$ must reach its minimum possible value.
\item[C4] If $\{\rho_{jj}=1/d\}_{j=0}^{d-1}$, then $P$ must reach its minimum value. In addition, if $\rho$ is pure, then $V$ must reach its maximum value.
\item[C5] If $\rho_{jj}>\rho_{kk}$ for some $(j,k)$, the value of $P$ cannot be increased by setting $\rho_{jj}\rightarrow\rho_{jj}-\epsilon$ and $\rho_{kk}\rightarrow\rho_{kk}+\epsilon$, for $\epsilon\in\mathbb{R}_{+}$ and $\epsilon\ll1$. And $V$ cannot be increased when decreasing $|\rho_{jk}|$ by an infinitesimal amount, for $j\ne k$.
\item[C6] $P$ and $V$ must be convex functions, i.e., $f(\sum_i \lambda_i \rho_i)\le \sum_i \lambda_i f(\rho_i)$, with $\sum_i \lambda_i = 1$, $\lambda_i \in [0,1]$, $f = P, V$ and for $\rho_i$ being valid density matrices.
\end{itemize}


\end{document}